\documentclass{amsart}
\usepackage{amssymb,graphicx,mathtools,amsthm,microtype,mathrsfs,crossreftools,color,amsmath,setspace,geometry}

\definecolor{darkblue}{rgb}{0.1,0,0.5}
\usepackage[linktocpage=true,colorlinks,linkcolor=darkblue,anchorcolor=green,citecolor=darkblue]{hyperref}

\usepackage{thmtools}
\usepackage[noabbrev,capitalise]{cleveref}

\theoremstyle{theorem}
\newtheorem{theorem}{Theorem}[section]
\newtheorem{lemma}[theorem]{Lemma}
\newtheorem{proposition}[theorem]{Proposition}
\newtheorem{corollary}[theorem]{Corollary}
\newtheorem{definition}[theorem]{Definition}

\newtheorem{claim}[theorem]{Claim}

\theoremstyle{remark}
\newtheorem{remark}[theorem]{Remark}
\AtEndEnvironment{remark}{\hfill$\blacktriangleleft$}

\DeclareMathOperator{\Telic}{Telic}

\DeclareMathOperator{\Id}{Id}

\DeclareMathOperator{\Per}{Per}
\DeclareMathOperator{\Size}{size}

\DeclareMathOperator{\Time}{Time}
\DeclareMathOperator{\One}{\textbf{1}}

\newcommand{\DD}{\mathbb{D}}

\newcommand{\NN}{\mathbb{N}}

\newcommand{\QQ}{\mathbb{Q}}
\newcommand{\RR}{\mathbb{R}}

\newcommand{\ZZ}{\mathbb{Z}}

\newcommand{\ccP}{\textsf{\upshape P}}
\newcommand{\ccNP}{\textsf{\upshape NP}}

\newcommand{\ccA}{\textsf{\upshape A}}

\newcommand{\ccM}{\textsf{\upshape M}}
\newcommand{\ccN}{\textsf{\upshape N}}
\newcommand{\ccK}{\textsf{\upshape K}}

\newcommand{\ccSAT}{\textsf{\upshape SAT}}

\newcommand{\ccFNP}{\textsf{\upshape FNP}}

\newcommand{\ccGRTELIC}{\textsf{\upshape G}\mathbb{R}\textsf{\upshape --TELIC}}
\newcommand{\ccRTELIC}{\mathbb{R}\textsf{\upshape --TELIC}}
\newcommand{\ccRBTELIC}{\mathbb{R}\textsf{\upshape --BTELIC}}
\newcommand{\ccRSTELIC}{\mathbb{R}\textsf{\upshape --STELIC}}
\newcommand{\ccRSBTELIC}{\mathbb{R}\textsf{\upshape --SBTELIC}}

\newcommand{\cA}{\mathcal{A}}
\newcommand{\cB}{\mathcal{B}}

\newcommand{\cG}{\mathcal{G}}

\newcommand{\cI}{\mathcal{I}}

\newcommand{\cQ}{\mathcal{Q}}

\newcommand{\cS}{\mathcal{S}}

\newcommand{\cU}{\mathcal{U}}

\newcommand{\cW}{\mathcal{W}}

\title[Correspondences in computational and dynamical complexity II]{Correspondences in computational and dynamical complexity II: forcing complex reductions}
\author[S. Everett]{Samuel Everett}
\address{University of Chicago}
\email{same@uchicago.edu}

\date{}

\begin{document}

\begin{abstract}
An \emph{algebraic telic problem} is a decision problem in $\textsf{NP}_\RR$ formalizing finite-time reachability questions for one-dimensional dynamical systems.
We prove that the existence of ``natural" mapping reductions between algebraic telic problems coming from distinct dynamical systems implies the two dynamical systems exhibit similar behavior (in a precise sense).
As a consequence, we obtain explicit barriers for algorithms solving algebraic telic problems coming from complex dynamical systems, such as those with positive topological entropy. For example, some telic problems cannot be decided by uniform arithmetic circuit families with only $+$ and $\times$ gates.
\end{abstract}

\maketitle
\tableofcontents
\pagebreak

\setstretch{1.06}

\section{Introduction}

This paper studies the computational complexity of finite-time reachability questions derived from dynamical systems, and asks how dynamical structure constrains the kinds of algorithms that can solve them. Informally, given a dynamical system and a target region of state space, we consider decision and search problems of the form: does there exist a suitably encoded initial state whose orbit reaches the target after a prescribed number of iterates? Our aim is to develop a technique for using concepts from dynamical systems to place \emph{a priori} restrictions on algorithms deciding such reachability problems.
The critical insight is that the definition of a reduction between two reachability problems coming from distinct dynamical systems is nearly identical to the definition of a conjugacy of the two dynamical systems. 

It is natural to derive computational problems from dynamical systems, and there is a sizable literature dedicated to connecting computability and complexity theory with dynamics. Salient examples include (i) determining the computability and complexity of state-reachability and related decision problems for a fixed system \cite{gracca2018computing,braverman2009constructing,braverman2006non,rojas2019computational,gracca2024robust}, and (ii) questions about when a dynamical system can simulate a universal Turing machine \cite{moore1991generalized,cotler2024computational,koiran1999closed,bournez2013computation,cardona2021constructing,cardona2022turing,cardona2024hydrodynamic}. (See the introduction of the first part of this work \cite{everett1} for a detailed overview.)

These developments have produced powerful computability and complexity-theoretic tools for analyzing dynamical systems. In contrast, it is comparatively less explored how \emph{dynamical} notions can be used to inform fundamental complexity-theoretic questions. The existing work in this direction has already led to deep insights (see, e.g.\ \cite{blanc_et_alICALP,bournez2017polynomial,cotler2024computational}), and motivates the program pursued here.

\subsection{Overview of main definitions and results}

For the purpose of this paper, a \emph{dynamical system} is a pair $(X,F)$ consisting of a set $X$ (the \emph{state space}) and a transformation $F:X\to X$. In most settings of interest, $X$ is a compact metric space and $F$ is continuous. To witness dynamical behavior, iterate $F$ on points $x\in X$, writing $F^n(x)$ for the $n$-fold composition of $F$ applied to $x$ (with $n\in\NN$, or $n\in\ZZ$ when $F$ is invertible).

In this paper we consider dynamical systems $(I, F)$ of the unit interval $I=[0,1]$ that are \emph{computable by Blum--Shub--Smale machines}: BSS-computable systems are those for which there exists a BSS-machine that on input $x \in I$, returns $F(x) \in I$. Issues of complexity do not arise in this case because the input $x \in I$ is always of length one.

An \emph{algebraic telic problem}, or simply a \emph{telic problem} coming from a dynamical system $(I, F)$ is a decision problem contained in $\ccNP_\RR$, and defined as follows.
Let $g:I\rightarrow I$ be a BSS-computable homeomorphism of the unit interval (e.g.\ a polynomial whose restriction to $I$ is a bijection).
Let $I_m=\{p/2^m:p\in [2^m]\}$ denote the precision $m$ dyadic rationals contained in $I$.
Then on length $n+2$ input $(a, b, 1,\dots,1)$ for $a, b \in I$ and $a\leq b$, decide whether there is a dyadic rational $y$ of precision $n^2$, i.e.\ $y \in I_{n^2}$, such that $F^n(g(y)) \in [a, b]$.
We denote an algebraic telic problem coming from a dynamical system $(I, F)$ by $\ccRTELIC_{(I, F)}(g)$.

We will also work with a generalization of algebraic telic problems called \emph{bounded algebraic telic problems}, defined as follows.
Let $(I, F)$ be a BSS-computable dynamical system and $g:I\rightarrow I$ a BSS-computable homeomorphism. Let $k:\NN\rightarrow \NN$ be a function computable by a polynomial-time Turing machine over a unary alphabet, and let $a:\NN\rightarrow [0,1]$ denote a function computable by a polynomial-time BSS machine.
Then on length $n+1$ input $(c, 1,\dots,1)$ for $c\in I$, decide whether there is a precision $n^2$ dyadic rational $y \in I_{n^2}$ such that $F^{k(n)}(g(y)) \in [c-a(n), c+a(n)]$. Bounded algebraic telic problems are denoted $\ccRBTELIC_{(I, F)}(a, k, g)$.

We find the most success in studying the search versions of telic problems, in which a dyadic rational $y$ witnessing $F^n(g(y)) \in [a, b]$ is returned. Search versions of standard algebraic telic problems and bounded algebraic telic problems are denoted
\[
\ccRSTELIC_{(I, F)}(g) \text{ and } \ccRSBTELIC_{(I, F)}(a,k,g),\text{ respectively}.
\]

A central idea of this paper is that writing down the form any reduction between telic problems must take comes very close to simply writing down the definition of a conjugacy of dynamical systems.
In particular, if a reduction can take a certain form, it implies the underlying dynamical systems share certain characteristics. 
As a consequence, if two telic problems are derived from systems which do not share those characteristics (are classified differently), reductions of that type cannot go through.

Although focusing on the existence of \emph{reductions} between problems is an unusual approach, it is the right frame of mind to use in our setting. Moreover, studying the structure of reductions between telic problems extends to the structure of algorithms solving the problems directly: considering reductions to telic problems coming from the trivial system $(I, \Id)$ is tantamount to considering algorithms solving the problems directly.

We now make these ideas concrete.

Let $\cI$ label the collection of all points and closed intervals contained in $I$, represented as the convex subset of $\RR^2$ whose points $(c,r)$ we correspond to intervals $[c-r,c+r]\subseteq I$.
Suppose the system $(I, F)$ has complicated orbit structure while $(I, T)$ is dynamically regular, so its search telic problems are solvable in polynomial time.
A natural method for solving $\ccRSTELIC_{(I, F)}(g)$ is to reduce it to some $\ccRSTELIC_{(I, T)}(\hat{g})$ coming from $(I, T)$ by constructing a family of maps $\{\eta_n\}_{n=1}^\infty$, with each $\eta_n:\cI \rightarrow \cI$ satisfying the property that for all $x \in I_{n^2}$ and $J \in \cI$, $F^n(g(x)) \in J$ if and only if $T^n(\hat{g}(x)) \in \eta_n(J)$.
We call such a reduction a \emph{level 3 natural search reduction from $\ccRSTELIC_{(I, F)}(g)$ to $\ccRSTELIC_{(I, T)}(\hat{g})$}.

We define stricter templates (levels 1--2) by imposing additional constraints on the functions $\{\eta_n\}$ composing the reduction,
and a more permissive template (level 4) that can be used to reduce between any two search problems, by only requiring the one-way implication
$T^n(\hat g(x))\in \eta_n(J)\Rightarrow F^n(g(x))\in J$.

Our first theorem shows that search telic problems coming from dynamically mismatched systems cannot be reduced to one another by natural reductions.

\begin{theorem}[\cref{thmNoNaturalSearch}]\label{thmNoNaturalSearch1}
There exists a dynamical system $(I, F)$ with positive topological entropy and a regular dynamical system $(I, T)$ with zero topological entropy, both computable by BSS-machines, as well as homeomorphisms $g:I\rightarrow I$ and $\hat{g}:I\rightarrow I$ computable by BSS-machines, for which $\ccRSTELIC_{(I, T)}(\hat{g}) \in \ccP_\RR$, and $\ccRSTELIC_{(I, F)}(g)$ is not reducible to $\ccRSTELIC_{(I, T)}(\hat{g})$ by level 1, 2, or 3 natural search reductions.
\end{theorem}

The guiding principle is that reductions between telic problems cannot ignore dynamical distinctions. If two systems are dynamically ``too far apart," then no \emph{simple} reduction can carry the computational burden from one to the other.

Our second theorem shows that even level 4 reductions are forced to be highly irregular when the two dynamical systems in question exhibit distinct behavior.
It is stated for bounded algebraic search telic problems. Let $\rho$ be the Euclidean metric on $\RR^2$, and let $\cI_{a}\subset\RR^2$ denote the set of points $(c, a)$ in the plane corresponding to intervals $[c-a,c+a]$ of fixed radius $a$.

\begin{theorem}[\cref{thm:fixedPoint}]\label{thm:fixedPoint1}
There exists a BSS-computable dynamical system $(I, F)$, and functions $a(n), k(n)$ for which $\ccRSBTELIC_{(I, F)}(a, k,\Id)$ is not reducible to the trivial search problem  $\ccRSBTELIC_{(I, \Id)}(a, \Id,\Id)$ coming from $(I, \Id)$ by level 4 natural search reductions whose functions $\eta_n:\cI_{a(n)}\rightarrow \cI_{a(n)}$ composing the sequence $\{\eta_n\}_{n=1}^\infty$ satisfy any of the following properties:
\begin{enumerate}
\item continuity,
\item non-expansivity,
\item there is a lower semi-continuous function $\varphi:\RR^2\rightarrow [0,\infty)$ such that
\[
\rho(z, \eta_n(z))\leq \varphi(z)-\varphi(\eta_n(z))
\]
for all $z \in \cI_{a(n)}$.
\end{enumerate}
\end{theorem}

The continuity restriction alone in \cref{thm:fixedPoint1} rules out level 4 reductions realized by uniform arithmetic-circuit families with only $+,\times$ gates.

Notice we do not place emphasis on studying the \emph{complexity} of computing the functions $\{\eta_n\}_{n=1}^\infty$ composing the reductions directly. Rather, we aim to determine what mathematical structures they must take; if enough complicated structure is forced upon the $\{\eta_n\}$, this will imply complexity lower bounds.
Additionally, the decision problems introduced in this paper are contained in $\ccNP_\RR$.
Hence, to place our results into a deeper perspective, telic problems add texture to the notion of \emph{complexity} in dynamical systems: separations in the complexity of solving telic problems associated with distinct dynamical systems imply the systems are distinguished at a deep, logical level extending beyond any standard topological or analytic differences they may have.

\subsection{Organization}
This paper is the second in a collection \cite{everett1,everett3} that seeks to develop a bidirectional relationship between dynamical systems and computational complexity theory. The papers use a common class of decision problems as the connective tissue between the disciplines. Moreover, the papers are each self-contained and may be read in any order.

In \cref{secPreliminaries} of this paper, we give necessary background in computational complexity and dynamical systems. \cref{secRealTelic} introduces algebraic telic problems, and gives initial classification and structure results. Following the definitions, in  \cref{secRealTelicReductions} and \cref{secSearchRed} we prove the main theorems of this paper.

\subsection*{Acknowledgements}
I would like to thank David Cash for the feedback and support throughout the rather long development of this project.
The author is partially supported by the National Science Foundation Graduate Research Fellowship Program under Grant No. 2140001.

\section{Preliminaries}\label{secPreliminaries}

\subsection{Background on dynamical systems}

We only state a few basic notions in this subsection, with less commonly used ideas presented as needed. If the reader wishes to obtain a more comprehensive background in dynamics, we suggest \cite{robinson1998dynamical,akin1993general}.

In the broadest sense, a \emph{dynamical system} is any group or semigroup action $\phi:G \times X \rightarrow X$. The actions can be discrete, continuous or smooth. Accordingly, there is a deep theory for general groups, both amenable and non-amenable, discrete, and topological. However, the general theory is based in the initial study of the well understood $\ZZ$ and $\RR$-actions and their associated semigroup actions, which are the focus of this paper.
In particular, whenever we say the pair $(X, T)$ is a \emph{dynamical system} or simply a \emph{system}, we mean $X$ is a compact metric space and $T$ is a continuous transformation. In this paper we will take $X = I= [0,1]$.

Two continuous transformations $T:X \rightarrow X$ and $S:Y \rightarrow Y$ are  \emph{topologically conjugate}, or just \emph{conjugate}, if there exists a homeomorphism $h:X \rightarrow Y$ such that $h\circ T = S \circ h$. We say $h$ is a \emph{topological semi-conjugacy} or \emph{morphism} from $T$ to $S$ provided $h$ is continuous, onto, and $h\circ T = S \circ h$.
If two systems are conjugate, they are considered essentially identical because most important measures of complexity (such as the entropy) are invariant under conjugacy.

\subsection{Background on real complexity theory and the Blum--Shub--Smale model}\label{secBSS}
For a proper source covering the basics of computational complexity theory see \cite{arora2009computational,papadimitriou}.
We begin by defining a Blum--Shub--Smale (BSS) machine \cite{blum1989theory}. We refer the reader to \cite{blum2012complexity} for a comprehensive introduction to BSS-machines, and the corresponding theory of real computation, while \cite{meer1997survey} gives a nice survey of real structural complexity theory.

Essentially, a Blum--Shub--Smale (BSS) machine is a Random Access Machine capable of performing the basic arithmetic operations at unit cost, with registers holding arbitrary real numbers. In more detail, a BSS-machine is a generalization of the standard Turing machine model with bit operations, and is instead capable of computing over arbitrary rings and fields. BSS-machines have a finite number of internal states, each of which corresponds to one of three types: (1) shift state: moving the head left or right, (2) branch state: changing states based on content of current cell, and (3) computation state: this state has a hard-wired function $f$ that reads the content $a$ of the current cell and replaces it with $f(a)$. When computing over $\RR$, $f$ is  a rational function of the form $g/h$ where $g$ and $h$ are polynomials.

For completeness, we give the formal definition of Blum--Shub--Smale machines, although we do not work closely with the definition. We use the clear definition of BSS-machines over $\RR$ given by Meer and Michaux in their lucid paper \cite{meer1997survey}.

\begin{definition}
Let
\[
Y \subset \RR^\infty\coloneqq \bigcup_{d \in \NN}\RR^d.
\]
A \emph{BSS-machine $\ccM$ over $\RR$ with admissible input set $Y$} is defined by a finite set $J$ of instructions labeled by $0,\dots,N$, and a finite set of \emph{machine constants} denoted $\alpha_q$.
A \emph{configuration} of $\ccM$ is a quadruple
\[
(n, i, j, x) \in J \times \NN \times \NN \times \RR^\infty,
\]
where, $n$ denotes the currently executed instruction, $i$ and $j$ are addresses, and $x$ is the content of the registers of $\ccM$. The initial configuration of $\ccM$'s computation on input $y \in Y$ is $(1, 1, 1, y)$. If $n = N$ and the actual configuration is $(N, i, j, x)$, the computation stops with output $x$. The instructions $\ccM$ can perform are as follows:
\begin{enumerate}
\item \textit{Computation:} $n:x_s \leftarrow x_k \circ_n x_l$ where $\circ_n$ is an operation such as $+,-,\times$, or $n :x_s \leftarrow \alpha_q$ for machine constant $\alpha_q \in \RR$. The register $x_s$ is assigned the value $x_k \circ_n x_l$ or $\alpha_q$ respectively. The other registers remain unchanged. The next instruction will be $n+1$, and the copy registers are changed according to $i \leftarrow i+1$ or $i \leftarrow 1$ and similarly for $j$.
\item \textit{Branching:} \emph{$n$: if $x_0 \geq 0$ goto $\beta(n)$ else goto $n+1$}, where $\beta(n) \in J$. The answer of the check determines the next instruction is determined.
\item \textit{Copying:} \emph{$n:x_i \leftarrow x_j$}. The ``read"-register is copied into the ``write"-register.
\end{enumerate}
\end{definition}

To any BSS-machine $\ccM$ over $Y$ corresponds the partial function $\Phi_\ccM:Y\rightarrow \RR^\infty$ computed by $\ccM$. This is called the \emph{input-output map}.
In order to deal with complexity theoretic issues for BSS-machines, we must define appropriate cost measures for computation.

\begin{definition}
Let $x \in \RR^\infty$ such that $x = (x_1,\dots,x_k,0,0,\dots)$. Then the \emph{size} of $x$ is defined as $\Size(x) \coloneqq k$. Then the size of the vector of $n$ ones is $\Size(\One_n) = n$.

Let $\ccM$ be a BSS-machine over $Y\subset \RR^\infty$, and $y \in Y$. The \emph{running time} of $\ccM$ on $y$ is defined by
\[
\Time_\ccM(y)\coloneqq \begin{cases}
\text{Number of operations executed by $\ccM$ if $\Phi_\ccM(y)$ is defined},\\
\infty \text{ otherwise}.
\end{cases}
\]
If $f:\NN\rightarrow \NN$, we say $\ccM$ computes $\Phi_\ccM$ in \emph{$f$-time} or \emph{in time $f$} if its computation on every input $x$, $\Time_\ccM(x) \leq f(\Size(x))$.
\end{definition}

We shall say  a machine $\ccM$ over $\RR$ is a \emph{polynomial time} machine on $Y \subset \RR^\infty$ if there are positive integers $c$ and $q$ such that $\Time_\ccM(y) \leq c(\Size(y))^q$ for all $y \in Y$. Then a map $\varphi:X\rightarrow Y \subset \RR^\infty$ is said to be a \emph{$p$-morphism} over $\RR$, or \emph{polynomial time computable}, if $\varphi$ is computable by a polynomial time machine on $X$.

Notice the definitions of size and run-time are independent of the magnitude of the real numbers contained in the ``tape cells" of the machine.
Moreover, the cost of any basic operation is 1, no matter what the operands are.
In fact this can lead to a number of peculiar results, breaking from the intuition of the Turing model; this has motivated the introduction of alternate measures of complexity (see the foregoing references).

The \emph{halting set} of $\ccM$ is the set of all inputs $y$ for which $\Phi_\ccM(y)$ converges (is defined).
A \emph{decision problem} over $\RR$ is a set $S \subseteq \RR^\infty$; its decidability/complexity is measured by the computability/complexity of its characteristic function $\chi_S$.
In keeping with tradition, we are most interested in \emph{structured} decision problems.
Let $A \subset B \subset \RR^\infty$.
A pair $(B, A)$ is called a \emph{structured decision problem}, and is \emph{decidable} if and only if there exists a BSS-machine $\ccM$ with admissible input set $B$, such that $\Phi_\ccM$ is the characteristic function of $A$ in $B$. Then $\ccM$ is said to \emph{decide} the decision problem $(B, A)$.

Notice if we suppose (as we do unless stated otherwise) that the set of problem instances $B$ is decidable over $\RR$, then $A$ is computable if and only if $\chi_{A}$ is. Thus, we freely use \emph{decision problem} to denote either a structured or an unstructured problem.

\begin{definition}
A decision problem $(B, A)$ with $A \subset B \subset \RR^\infty$ belongs to the class $\ccP_\RR$ (\emph{deterministic polynomial time}) if and only if there exists a BSS-machine $\ccM$ with admissible input set $B$ and constants $k \in \NN$, $c \in \RR$ such that $\ccM$ decides $(B, A)$ and for all $y \in B$, $\Time_\ccM(y) \leq c\cdot \Size(y)^k$.
\end{definition}

\begin{definition}
A decision problem $(B, A)$ with $A \subset B \subset \RR^\infty$ belongs to the class $\ccNP_\RR$ (\emph{nondeterministic polynomial time}) if and only if there exists a BSS-machine $\ccM$ with admissible input set $B \times \RR^\infty$ and constants
$k \in \NN$,
$c \in \RR$
such that the following conditions hold:
$\Phi_\ccM(y, z) \in \{0,1\}$, $\Phi_\ccM(y, z) =1 \implies y \in A$, and
\[
\forall y \in A, \exists z \in \RR^\infty, \text{ such that } \Phi_\ccM(y, z) = 1 \text{ and } \Time_\ccM(y, z) \leq c \cdot \Size(y)^k.
\]
\end{definition}

The nondeterminism in the definition of $\ccNP_\RR$ refers to the vector $z$---it is the certificate or ``guess" that proves $y \in A$ in polynomial time (there is no indication of how to obtain the certificate).
Similarly to the classical setting, $\ccP_\RR \subseteq \ccNP_\RR$, and it is a central question whether $\ccP_\RR = \ccNP_\RR$. In addition, all problems in $\ccNP_\RR$ are solvable in single exponential time \cite{blum1989theory}.

We conclude this section be introducing the crucial concept of a reduction between decision problems.

\begin{definition}
Let $(B_1, A_1)$ and $(B_2, A_2)$ be decision problems.
We shall say $(B_2, A_2)$ is \emph{reducible in polynomial time} to $(B_1, A_1)$ if and only if there exists a BSS-machine $\ccM$ over $B_2$ such that $\Phi_\ccM(B_2) \subset B_1$, $\Phi_\ccM(y) \in A_1 \iff y \in A_2$, and $\ccM$ halts in polynomial time. This is denoted as $(B_2, A_2) \leq_\RR (B_1,A_1)$. A decision problem $(B, A)$ in $\ccNP_\RR$ is \emph{$\ccNP_\RR$-complete} if and only if every other problem in $\ccNP_\RR$ reduces to $(B, A)$ in polynomial time.
\end{definition}

Although search problems in the BSS model do not have the same meaning as in the Turing model on account of computability issues, we successfully work with the concept, giving a definition in the body of the paper.

\section{Algebraic telic problems}\label{secRealTelic}

Let $\DD$ denote the set of \emph{dyadic rational numbers}; the numbers $y \in \DD$ have form $y = p/2^r$ for some integers $p$ and $r \geq 0$.
Let $\DD_r = \{m \cdot 2^{-r} : m \in \ZZ\}$ denote the collection of dyadic rationals with precision $r$.
If $J \subset \RR$, let $J_r = J \cap \DD_r$. We call $J_r$ the \emph{$r$-discretization of the set $J$}.

Let $I = [0,1]$, and let $F:I\rightarrow I$ be a transformation \emph{computable by a BSS-machine} $\ccM$, so $\Phi_\ccM(x) = F(x)$ for all $x \in [0,1]$.
Since every computation step of a BSS-machine computes a rational function $f/g$ of the operands, where $f$ and $g$ are hard-wired polynomials, and branching is allowed, it is clear that standard dynamical systems of the interval, such as quadratic maps, piecewise continuous dynamical systems such as the tent map, as well as other systems like expanding maps and rigid circle rotations (implemented as $2x \mod 1$ or $x+\alpha \mod 1$) can all be computed by BSS-machines.
We say such systems $(I, F)$ are \emph{computable by BSS-machines}.\footnote{The input is always a single element of $\RR$, so all inputs are the same size. This avoids the need to use terms like ``efficient" when describing such machines: there is no scaling in the complexity of computing single iterations of the maps.}

More generally, we will say a function $f:X\rightarrow Y$, $X \subseteq \RR^n$, $Y \subseteq \RR^m$, is \emph{computable by a BSS-machine} if there exists a BSS-machine $\ccM$ such that $\Phi_\ccM(x) = f(x)$ for all $x \in X$.
Here $I_r = \DD_r \cap [0,1]$ denotes the $r$-discretization of the unit interval.

We now define \emph{algebraic telic problems}. Recalling the intuition from the introduction, a telic problem is essentially a reachability problem asking whether some collection of points intersects with a target subset in state space after the dynamical system has evolved for $n$ time-steps.

\begin{definition}[Algebraic telic problems]\label{defRtelic}
Let $(I, F)$ be a dynamical system computable by a BSS-machine $\ccM$. Let $\cB\subset \RR^\infty$ be the set $\cB = \cup_{n=1}^\infty \cB^{(n)}$ where 
\[
\cB^{(n)} = [0,1]\times [0,1] \times \prod_{j=1}^n \{1\},
\]
each of which is interpreted as a subset of $\RR^{n+2}$; if $y \in \cB^{(n)}$ then $\Size(y) = n+2$.
Let $g:I \rightarrow I$ be a homeomorphism of $I$ computable by a BSS-machine, and define
\[
\cA^{(n)} = \left\{(a, b, 1,\dots,1) \in \cB^{(n)} : a\leq b, \text{ and } \exists s \in I_{n^2} \text{ s.t. } F^n(g(s)) \in [a, b]\right\}.\footnote{The choice of $n^2$ for $I_{n^2}$ is arbitrary. Rather than $n^2$ we can take any efficiently computable strict monotone increasing function $l(n)\geq n$ most generally.}
\]
Let $\cA = \cup_{n=1}^\infty \cA^{(n)}$.
An \emph{algebraic telic problem} coming from a BSS-computable one-dimensional dynamical system $(I, F)$, is the structured decision problem $\ccRTELIC_{(I, F)}(g) \coloneqq (\cB, \cA)$.\footnote{The word \emph{telic} comes from the Greek word ``telos" for end/goal, or ``reaching a goal," which aptly characterizes the problems considered in this paper as a class of reachability problem.}
\end{definition}

Let $\cG$ be the class of all homeomorphisms of the unit interval computable by BSS-machines. Put
\[
\Telic_{(I, F)} = \bigcup_{g \in \cG}\ccRTELIC_{(I, F)}(g).
\]

\begin{lemma}\label{lemRealTelicNP}
Let $(I, F)$ be a BSS-computable system. Then $\Telic_{(I, F)} \subset \ccNP_\RR$.
\end{lemma}
\begin{proof}
We construct a verifier BSS-machine $\ccN$ for a fixed algebraic telic problem $\ccRTELIC_{(I, F)}(g)$, $g \in \cG$.

Let $\ccN$ be a BSS-machine with admissible input set $\cB \times \RR^\infty$ that operates as follows. On input $(x, y) \in \cB \times \RR^\infty$, begin by checking if $y \in I_{n^2}$ where $n = \Size(x)-2$; this can be checked in $O(n^2)$ time by the following standard procedure: first check whether $y < 0$ or $y > 1$ and halt if so. Otherwise continue, setting $z \leftarrow y$ and loop $n$ times performing the following operations: $z\leftarrow 2z$, then if $z \ge 1$ put $z \leftarrow z-1$. After the loop has terminated, return $1$ if $z=0$ and return 0 (``No") otherwise.

If $y$ is not in the proper form, reject. Otherwise, continue by computing $F^n(g(y))$, and checking if $a \leq F^n(g(y)) \leq b$---clearly these are polynomial-time operations in $n$ since computing $g(y)$ does not scale with $n$, and computing $F(x)$ does not scale with $n$, so computing $F^n$ involves computing $n$ iterations of $F$. If ``yes" and $F^n(g(y)) \in [a, b]$, accept, and if not reject and halt, an $O(n)$ operation. Thus $\ccRTELIC_{(I, F)}$ belongs to the class $\ccNP_\RR$.
\end{proof}

There are many systems $(I, F)$ for which $\Telic_{(I, F)} \subset \ccP_\RR$. These are the ``regular" systems one would expect to admit easy telic problems.

\begin{proposition}\label{propRotationInP}
Let $R_\alpha(x) = x+\alpha \mod 1$, $\alpha \in (0,1)$ be a rigid rotation of the circle.
Then $\ccRTELIC_{(S^1, R_\alpha)}(g) \in \ccP_\RR$ for all $g \in \cG$.
\end{proposition}
\begin{proof}
Proof of this assertion is given in full rigorous detail in \cite{everett1}. However, the proof idea is simple and can be succinctly stated as follows. Notice that since $g$ is either increasing or decreasing, it preserves or reverses the order of $I$ with respect to $\leq$: that is, either $a\le b\implies g(a)\le g(b)$, or $a\le b\implies g(b)\le g(a)$. Moreover, the map $R_\alpha$ shares this property for every $\alpha \in [0,1)$. Hence, determining whether there is an $s \in I_{n^2}$ such that $R_\alpha^n(g(s)) \in [a, b]$ reduces to performing a binary search over the space of precision $n^2$ dyadic rationals $I_{n^2}$: pick $s \in I_{n^2}$, compute $R_\alpha^n(g(s))$, and branch based on which side of the target interval $J = [a, b]$ the point $R_\alpha^n(g(s))$ lies on.
\end{proof}

It is often the case in the Boolean model of computation that search and decision problems are closely connected or equivalent in an appropriate sense, i.e.\ through search to decision reductions as found in $\ccSAT$.
However, when working with real numbers at perfect precision search problems cannot be expected to be computable in general---instead one must study approximation algorithms. One example is finding zeros of polynomials. This generally leads to issues of numerical analysis; see \cite{blum2012complexity} for discussion.

On the other hand, the search version of telic problems is clearly computable, since a certificate is just a dyadic rational taken from some finite set of candidates. Moreover, we find it useful to work with search (function) problems when studying complexity theoretic issues of telic problems. For this reason we take a moment to formalize this notion.

A \emph{search problem} over $\RR$ is a polynomially bounded (in the size of inputs) and polynomially decidable (by a BSS-machine) relation  $R\subset \RR^\infty \times \RR^\infty$. As in the classical setting, given an input $x \in \RR^\infty$, the task is to find a $y \in \RR^\infty$ such that $(x, y) \in R$, and return ``No" if no such $y$ exists. Let $\ccFNP_\RR$ denote the class of all polynomially bounded relations $R$ that are BSS-decidable in polynomial time. We let $\ccRSTELIC_{(I, F)}(g) \subset \RR^\infty \times \RR^\infty$ denote the search version of the telic problem $\ccRTELIC_{(I, F)}(g)$.

Reductions between search (function) problems in $\ccFNP_\RR$ are defined identically to the classical case: we say a search problem $\cQ$ reduces to a search problem $\cW$ if there are functions $R$ and $S$, both computable in polynomial time by a BSS-machine, such that for any inputs $x, z \in \RR^\infty$, if $x$ is an instance of $\cQ$, then $R(x)$ is an instance of $\cW$, and if $z$ is a correct output of $R(x)$, then $S(z)$ is a correct output of $x$.

An immediate corollary of \cref{lemRealTelicNP} is
\begin{corollary}
Let $(I, F)$ be an efficiently computable system. Then $\ccRSTELIC_{(I, F)}(g) \in \ccFNP_\RR$ for all $g \in \cG$.
\end{corollary}

\section{Natural reductions between algebraic telic problems}\label{secRealTelicReductions}

In this section we show that natural classes of reductions between decision versions of algebraic telic problems can be entirely ruled out. We interpret these results as rigorous means of saying that reductions between certain telic problems are forced to take a complicated form, since reductions with simple or obvious structure do not exist. This section serves as a prerequisite to \cref{secSearchRed}, in which we obtain stronger results in the case of search versions of algebraic telic problems.

\begin{remark}
Emphasizing the study of \emph{reductions} between telic problems is most appropriate due to the very well-developed classification theory of dynamical systems, and the fact that---as will be demonstrated below---writing down the form any reduction between telic problems must take comes very close to simply writing down the definition of a conjugacy of dynamical systems (recall two systems are conjugate if their dynamics are equivalent up to homeomorphism).

By showing that any ``natural" reduction between telic problems in fact forces the associated dynamical systems to share dynamical properties, we obtain an indirect method of proving that algorithms solving certain telic problems are forced to take certain forms.
In particular, notice that reductions to the trivial system $(I, \Id)$ correspond with algorithms simply solving the problem directly.
\end{remark}

\subsection{Ruling out natural telic reductions}

Suppose $(I, F)$ is a dynamical system computable by a BSS-machine, and $(I, F)$ has a complicated orbit structure so deciding $\ccRTELIC_{(I, F)}(g) \in \Telic_{(I, F)}$ may be challenging. Perhaps the most obvious approach in attempting to decide the problem efficiently is to find a more regular dynamical system $(I, T)$ so that instances of $\ccRTELIC_{(I, F)}(g)$ can be efficiently transformed into instances of some telic problem $\ccRTELIC_{(I, T)}(\hat{g})$ and solved in the nicer dynamical setting.

A most basic first attempt at seeing such a reduction through would involve constructing a function family $\{\psi_n\}_{n=1}^\infty$ with $\psi_n:I\rightarrow I$ satisfying the property that, for any $x \in I$ and $J=[a, b] \subseteq I$, $F^n(g(x)) \in J$ if and only if $T^n(\hat{g}(x)) \in \psi_n(J)$.
If such a function family $\{\psi_n\}$ were successfully constructed and the $\psi_n$ were computable in time polynomial in $n$ by a BSS-machine, then it could be used as a reduction from $\ccRTELIC_{(I, F)}(g)$ to $\ccRTELIC_{(I, T)}(\hat{g})$.

We call reductions of this form \emph{level 1 natural telic reductions}.
As the name suggests, we define a hierarchy of natural telic reductions of increasing complexity for both decision and search telic problems. We list what we deem to be the foundational natural telic reductions below, up to the most general form that must include any reduction between telic problems, and then supply a result ruling out the existence of some of these classes of reduction in certain cases.

\begin{remark}
Reductions by function families $\{\psi_n\}_{n=1}^\infty$ bring to mind reductions by non-uniform circuit families such as algebraic computation trees. However, we take these families to be uniformly computable.
\end{remark}

Now suppose the original condition on the $\psi_n$ in level 1 natural telic reductions is weakened so that, in reducing $\ccRTELIC_{(I, F)}(g)$ to $\ccRTELIC_{(I, T)}(\hat{g})$ we instead construct a function family $\{\psi_n\}_{n=1}^\infty$ with $\psi_n:I\rightarrow I$ satisfying the property that for any $J=[a, b] \subseteq I$, there is an $x \in I_{n^2}$ such that $F^n(g(x)) \in J$ if and only if there is a $y \in I_{n^2}$ such that $T^n(\hat{g}(y)) \in \psi_n(J)$. We call reductions of this type \emph{level 2 natural telic reductions}.

Let $\cI$ denote the space of all points and closed intervals contained in the unit interval $I$. We will identify $\cI$ with the closed triangle
\[
\{(c, r) \in \RR^2 : 0 \leq c \leq 1, 0\leq r, r \leq \min\{c, 1-c\}\} \subset \RR^2,
\]
so that each element $(c, r) \in \cI$ describes a closed interval for which $c$ is the midpoint and $r$ is the radius: $[c-r, c+r]$.
We treat functions $f:\cI\rightarrow \cI$ of $\cI$ as functions of the set $\cI \subset \RR^2$. In addition, we abuse notation and write $x \in J=(c,r) \in \cI$ to mean $x \in I$ is contained in the closed interval $[c-r,c+r]$.

Now consider function families $\{\psi_n\}_{n=1}^\infty$ with $\psi_n:\cI\rightarrow \cI$ (treated as functions of $\cI \subset \RR^2$) such that for any $J\in \cI$, there is an $x \in I_{n^2}$ such that $F^n(g(x)) \in J$ if and only if there is a $y \in I_{n^2}$ such that $T^n(\hat{g}(y)) \in \psi_n(J)$; that is, $F^n(g(I_{n^2})) \cap J \neq \emptyset \iff T^n(\hat{g}(I_{n^2})) \cap \psi_n(J) \neq \emptyset$. We shall call reductions using such function families \emph{level 3 natural telic reductions} from $\ccRTELIC_{(I, F)}(g)$ to $\ccRTELIC_{(I, T)}(\hat{g})$.
In fact it is not hard to see that level 3 natural telic reductions are ``enough," in the sense that any telic problem $\ccRTELIC_{(I, F)}(g)$ can be reduced to some telic problem $\ccRTELIC_{(I, T)}(\hat{g})$ by some level 3 natural telic reduction. Although, the reduction may not be computable by polynomial-time BSS-machines.

Even so, there exists the following most general form of reduction between algebraic telic problems.
\emph{Level 4 natural telic reductions} from $\ccRTELIC_{(I, F)}(g)$ to $\ccRTELIC_{(I, T)}(\hat{g})$ are induced by function families $\{\psi_n\}_{n=1}^\infty$ and $\{\omega_n\}_{n=1}^\infty$ with $\psi_n:\cI\rightarrow \cI$ and $\omega_n: \cI \rightarrow \NN$ such that for any $J \in \cI$, there is an $x \in I_{n^2}$ such that $F^n(g(x)) \in J$ if and only if there is a $y \in I_{(\omega_n(J))^2}$ such that $T^{\omega_n(J)}(\hat{g}(y)) \in \psi_n(J)$.

Any reduction from some $\ccRTELIC_{(I, F)} \in \Telic_{(I, F)}$ to some $\ccRTELIC_{(I, T)} \in \Telic_{(I, F)}$ sits at one of the four levels of natural telic reductions.
Furthermore, notice that each level of natural telic reductions is nontrivial, and are not so restrictive to be vacuous.
Indeed, it is not hard to identify nontrivial systems for which their associated telic problems are reducible to one another by level one natural telic reductions.

\begin{remark}
It may seem natural to augment many of the levels of natural telic reductions with additional function families $\{\sigma_n\}$ so that, for instance, we have $F^n(x) \in J$ if and only if $T^n(\sigma_n(x)) \in \psi_n(J)$. This cannot be done however unless all the $\sigma_n$ are identical: if not, this is no-longer a reduction to a telic problem since algebraic telic problems $\ccRTELIC_{(I, F)}(g)$ assume the homeomorphism $g$ is fixed, and cannot take form $\sigma_n\circ g$ for different $\sigma_n$ on each new instance $n$.
\end{remark}

The following theorem asserts that the level 1 and 2 natural reduction templates between telic problems cannot be used when constructing reductions between telic problems coming from sufficiently distinct dynamical systems.

\begin{theorem}\label{thmRulingOutNaturalTelic}
There exists dynamical systems $(I, F)$ and $(I, T)$, both computable by BSS-machines, such that there do not exist level 1 or 2 natural telic reductions from $\ccRTELIC_{(I, F)}(\Id)$ to $\ccRTELIC_{(I, T)}(\Id)$.
\end{theorem}
\begin{proof}
Let $(I, F)$ be the tent-map system $F(x) = 2\min\{x, 1-x\}$, and let $(I, T)$ be a homeomorphism of the unit interval $T(x) = x^2$.

We begin by showing a level 1 natural telic reduction from $\ccRTELIC_{(I, F)}(\Id)$ to $\ccRTELIC_{(I, T)}(\Id)$ does not exist. Suppose there existed a family $\{\psi_n\}_{n=1}^\infty$ of maps $\psi_n:I\rightarrow I$ such that for any $x \in I$ and closed interval or point $J \subseteq I$, $F^n(x) \in J$ if and only if $T^n(x) \in \psi_n(J)$. Then fix an $n \geq 1$ and put $J = y \in I$ a point. Then $F^n(x) = y \iff T^n(x) = \psi_n(y)$. Hence, putting $\psi_n(F^n(x)) = \psi_n(y) = T^n(x)$, we conclude that $\psi_n\circ F^n = T^n$ since the property holds for all $y \in I$. But the tent-map $F$ is two-to-one while $T$ is one-to-one, so there exists $x_1\neq x_2$ so that $F(x_1) = F(x_2)$, but this implies $\psi_n(F(x_1)) = \psi_n(F(x_2)) = T(x_1)=T(x_2)$, a contradiction. Thus no such function family $\{\psi_n\}$ can exist, and level one natural telic reductions from $\ccRTELIC_{(I, F)}(\Id)$ to $\ccRTELIC_{(I, T)}(\Id)$ do not exist.

We now show impossibility of level 2 natural telic reductions. We may rewrite the definition of level 2 natural telic reductions into the following equivalent form: one telic problem is reducible to another telic problem by level 2 natural telic reductions if there is a function family $\{\psi_n\}$ with $\psi_n:I \rightarrow I$ such that for all $J =[a, b] \subseteq I$, $a\leq b$,
\[
J \cap F^n(I_{n^2}) \neq \emptyset \iff J \cap \psi_n^{-1}(T^n(I_{n^2}))\neq\emptyset \text{ for all } n \geq 1.
\]
We need the following elementary claim.
\begin{claim}\label{claimLvl2}
If $A, B \subseteq I$, and for all closed intervals and points $J \subseteq I$, $A \cap J \neq \emptyset \iff B \cap J \neq \emptyset$ then $A = B$.
\end{claim}
\begin{proof}[Proof of claim \cref{claimLvl2}]
Since $J$ is allowed to be a point, for any $x \in A$ take $J=[x,x]$, so by hypothesis $A \cap J \neq \emptyset$ if and only if $B \cap J \neq \emptyset$, hence $x \in B$. Thus $A \subseteq B$, and by a symmetric argument $B \subseteq A$, so $A=B$.
\end{proof}
Now, due to \cref{claimLvl2}, we have $F^n(I_{n^2}) = \psi_n^{-1}(T^n(I_{n^2}))$, and thus $\psi_n(F^n(I_{n^2})) = T^n(I_{n^2})$. But $T(x) = x^2$ is a bijection while $F(x) = 2\min\{x, 1-x\}$ is two-to-one and each iteration carries dyadic rationals on opposite sides of $1/2$ to the same point. Hence, we cannot have $\psi_n(F^n(I_{n^2})) = T^n(I_{n^2})$ since $|\psi_n(F^n(I_{n^2}))|< 2^{n^2}$ while $|T^n(I_{n^2})| = 2^{n^2}$, implying $\ccRTELIC_{(I,F)}(\Id)$ is not reducible to $\ccRTELIC_{(I, T)}(\Id)$ by level 2 natural telic reductions.
\end{proof}

The proof of \cref{thmRulingOutNaturalTelic} used only elementary concepts and simple dynamical systems. In our perspective this should be viewed positively: that level 1 and 2 reductions can be ruled out in the first place is not necessarily obvious, and the fact that this was accomplished using such elementary arguments suggests ruling out higher level reductions may be approachable using more sophisticated but already developed tools and techniques.
More may also be said about ruling out reductions between telic problems defined for a more general class of dynamical systems of higher-dimensional state spaces. See \cref{secGeneralReductions} for a general result proving that dynamical systems which do not share basic dynamical properties also fail to admit telic problems reducible to one another by natural reductions.

\section{Ruling out natural search reductions}\label{secSearchRed}

While the search version of natural telic reductions shares many similarities with the decision version, the additional requirement in search problems that a solution be returned is enough to force stronger results than in the decision case.

Let $(I, F)$ and $(I, T)$ be dynamical systems computable by a BSS-machine, and let $\ccRSTELIC_{(I, F)}(g)$ and $\ccRSTELIC_{(I, T)}(g)$ be associated algebraic search telic problems with homeomorphism $g$ in common. Solving an instance $(a, b, 1,\dots,1)$ of $\ccRSTELIC_{(I, F)}(g)$ involves returning an $s \in I_{n^2}$ such that $F^n(g(s)) \in [a, b]$, if one exists.
If $\ccRSTELIC_{(I, F)}(g)$ reduces to $\ccRSTELIC_{(I, T)}(g)$ in polynomial time, then there are functions $R$ and $S$ computable by a polynomial time BSS-machine, such that for any inputs $x, z \in \RR^\infty$, if $x$ is an instance of $\ccRSTELIC_{(I, F)}(g)$ then $R(x)$ is an instance of $\ccRSTELIC_{(I, T)}(g)$, and if $z$ is a correct output of $R(x)$ then $S(z)$ is a correct output of $x$.

In constructing such reductions, it is natural to define $R$ so that it carries an instance $x=(a, b, 1,\dots,1)$ of size $n+2$ from $\ccRSTELIC_{(I, F)}(g)$ to an instance $R(x)=(a',b', 1,\dots,1)$ also of size $n+2$ from $\ccRSTELIC_{(I, T)}(g)$. That is, map to instances involving the same number of iterations of the dynamical system.
But $S:I_{n^2} \rightarrow I_{n^2}$ for all $n \in \NN$, so it is natural to set $S \coloneqq \Id$ so that $R$ is now a function taking instances $x$ from $\ccRSTELIC_{(I, F)}(g)$ to instances $R(x)$ from $\ccRSTELIC_{(I, T)}(g)$ with the property that if $z$ is the solution to $R(x)$ then $z$ is the solution to instance $x$.
In addition, since the function $R(x)$ leaves the length of an instance invariant, for fixed $n$, it may be reinterpreted as a function taking closed intervals to closed intervals. Hence, we can identify $R(x)$ with a family of  functions $\{\eta_n\}_{n=1}^\infty$ with $\eta_n:\cI \rightarrow \cI$.

This leaves us in a situation similar to the analysis of reductions between decision telic problems: we may place restrictions on the functions $\eta_n$ and determine when reductions will go through.

Suppose $\ccRSTELIC_{(I, F)}(g)$ is reducible to $\ccRSTELIC_{(I, T)}(g)$ using a family of functions $\{\eta_n\}_{n=1}^\infty$ of the unit interval, $\eta_n:I\rightarrow I$, such that for all closed intervals $J \in \cI$ and all $x \in I_{n^2}$, $F^n(g(x)) \in J$ if and only if $T^n(g(x)) \in \eta_n(J)$. We call reductions using function families $\{\eta_n\}$ of this type \emph{level 1 natural search reductions}.

\begin{remark}
Level 1 natural search reductions restrict to reductions between search telic problems $\ccRSTELIC_{(I, F)}(g)$ and $\ccRSTELIC_{(I, T)}(g)$ with the same homeomorphism $g$, because considering different homeomorphisms $g$ and $\hat{g}$ is essentially equivalent. To see this, first note that $(g^{-1} \circ F \circ g)^n = g^{-1}\circ F^n \circ g$, and that if for all closed intervals $J \subseteq I$ and all $x \in I_{n^2}$, we have $F^n(g(x)) \in J$ if and only if $T^n(g(x)) \in \eta_n(J)$, then since $J$ can be a degenerate interval (a point), we have $\eta_n(F^n(g(x))) = T^n(g(x))$.
Put $\tilde{F} = g^{-1}\circ F \circ g$ and $\tilde{T} = \hat{g}^{-1} \circ T \circ \hat{g}$, and define $\tilde{\eta}_n = \hat{g}^{-1} \circ \eta_n \circ g$. Then using the above facts we have
\[
\tilde{\eta}_n\circ \tilde{F}^n = \hat{g}^{-1} \circ \eta_n\circ F^n\circ g = \hat{g}^{-1} \circ T^n\circ \hat{g} = \tilde{T}^n
\]
So we are left with $\tilde{\eta}_n(\tilde{F}^n(x)) = \tilde{T}^n(x)$ as when the two telic problems used identical homeomorphisms $g$.
\end{remark}

Moving up a level of complexity, we drop the restriction that the $\eta_n$ are transformations of the unit interval, and instead let them be functions of the space of closed intervals $\cI$. If $\ccRSTELIC_{(I, F)}(g)$ is reducible to $\ccRSTELIC_{(I, T)}(\hat{g})$ by a function family $\{\eta_n\}_{n=1}^\infty$ such that the $\eta_n:\cI \rightarrow \cI$ satisfy the property that for all $x \in I$ and $J \in \cI$, $F^n(g(x)) \in J$ if and only if $T^n(\hat{g}(x)) \in \eta_n(J)$, then we shall say $\ccRSTELIC_{(I, F)}(g)$ is reducible to $\ccRSTELIC_{(I, T)}(\hat{g})$ by \emph{level 2 natural search reductions}.

Level 3 natural search reductions will be similarly defined, only we take $x \in I_{n^2}$ instead of $I$: if $\ccRSTELIC_{(I, F)}(g)$ is reducible to $\ccRSTELIC_{(I, T)}(\hat{g})$ by a function family $\{\eta_n\}_{n=1}^\infty$, $\eta_n:\cI \rightarrow \cI$, such that for all $x \in I_{n^2}$ and $J \in \cI$, $F^n(g(x)) \in J$ if and only if $T^n(\hat{g}(x)) \in \eta_n(J)$, then we shall say $\ccRSTELIC_{(I, F)}(g)$ is reducible to $\ccRSTELIC_{(I, T)}(\hat{g})$ by \emph{level 3 natural search reductions}.

Level 4 natural search reductions are at a level of generality such that any search telic problem $\ccRSTELIC_{(I, F)}(g)$ is reducible to any other $\ccRSTELIC_{(I, T)}(\hat{g})$ via a level 4 natural search reduction, although perhaps not in polynomial-time. We say $\ccRSTELIC_{(I, F)}(g)$ is reducible to $\ccRSTELIC_{(I, T)}(\hat{g})$ by \emph{level 4 natural search reductions} if there exists a family of functions $\{\eta_n\}_{n=1}^\infty$, $\eta_n:\cI\rightarrow \cI$, such that for all $J \in \cI$ and for all $x \in I_{n^2}$, $T^n(\hat{g}(x)) \in \eta_n(J)$ implies $F^n(g(x)) \in J$, and $T^n(\hat{g}(x)) \not\in \eta_n(J)$ for all $x \in I_{n^2}$ implies $F^n(g(x)) \not\in J$ for all $x \in I_{n^2}$.

The subtle distinction between level 3 and level 4 natural search reductions is that the problems in level 3 reductions must have the same number of solutions per instance, while this need not be the case for level 4 natural search reductions. For example, $\ccRSTELIC_{(I, T)}(\hat{g})$ may have only one solution while $\ccRSTELIC_{(I, F)}(g)$ may have many, but so long as the solution to an instance of $\ccRSTELIC_{(I, T)}(\hat{g})$ is also a solution to the corresponding instance of $\ccRSTELIC_{(I, F)}(g)$, level 4 reductions go through.

The following theorem asserts there are dynamical systems admitting search telic problems for which level 1, 2, and 3 natural search reductions do not exist. As in the case of the decision problems, the level 1 and 2 natural search reductions specifically fail to work in certain cases for rather elementary reasons. Even so, this is perhaps not \emph{a priori} expected since neither the level 1 nor the level 2 natural search reductions are trivial---indeed the $\eta_n:\cI\rightarrow \cI$ in level 2 natural search reductions are arbitrary functions from points and intervals to points and intervals. As such, the following theorem supports the intuition that any efficient reduction from a telic problem coming from a chaotic system to a telic problem coming from a regular system will certainly not be trivial.

\begin{theorem}\label{thmNoNaturalSearch}
There exists a dynamical system $(I, F)$ with positive topological entropy and a regular dynamical system $(I, T)$ with zero topological entropy, both computable by BSS-machines, as well as homeomorphisms $g:I\rightarrow I$ and $\hat{g}:I\rightarrow I$ computable by BSS-machines, for which $\ccRSTELIC_{(I, T)}(\hat{g}) \in \ccP_\RR$, and $\ccRSTELIC_{(I, F)}(g)$ is not reducible to $\ccRSTELIC_{(I, T)}(\hat{g})$ by level 1, 2, or 3 natural search reductions.
\end{theorem}

We require the following preparatory lemmas and notation. Recall the expanding map $E_2:I\rightarrow I$ is defined as $E_2(x) = 2x\mod 1$.\footnote{Technically the map is of the circle $S^1$, not the unit interval, but these two situations are identical for our analysis.}

\begin{lemma}\label{lem:2nToOne}
Let $E_2:I\rightarrow I$ be the expanding map. Then for all but finitely many $x \in I$, the cardinality of $(E^n_2)^{-1}(x)\subset I$ is $2^n$ for all $n \geq 1$.
\end{lemma}
\begin{proof}
Proof follows immediately from observing that $E^n_2(x) = 2^nx \mod 1$.
\end{proof}

Fix the following piecewise linear function of the interval $\alpha:I\rightarrow I$, defined by
\[
\alpha(x) = \begin{cases}
\frac{1}{\sqrt{2}}x &\text{whenever } 0\leq x \leq 2-\sqrt{2}\\
\sqrt{2}x + (1-\sqrt{2}) &\text{whenever } 2-\sqrt{2}< x \leq 1.
\end{cases}
\]
The function $\alpha$ is a BSS-computable homeomorphism with the following property, used in proving \cref{lem:telicPerturb} below.

\begin{lemma}\label{lem:smallPreimage}
For every $k, n \geq 1$ and $s \in I_{n}$, there is no $s' \in I_n$ with $s'\neq s$ such that $E_2^k(\alpha(s))=E_2^k(\alpha(s'))$.
\end{lemma}
\begin{proof}
Recall $E_2^k(x) = 2^k(x) \mod 1$.
Let $\{z\}$ denote the fractional part of $z \in \RR$ (i.e.\ $z \mod 1$). Then for every $k, n \geq 1$, we aim to show there are no distinct $s, s' \in I_n$ such that $\{2^k\alpha(s)\} = \{2^k\alpha(s')\}$.

Set $s=\frac{a}{2^{n}},\ s'=\frac{b}{2^{n}}$ with $0\le a,b\le 2^{n}-1$.
If $\{2^{k}\alpha(s)\}=\{2^{k}\alpha(s')\}$, then $2^{k}\big(\alpha(s)-\alpha(s')\big)\in\ZZ$, which implies $\alpha(s)-\alpha(s')\in 2^{-k}\ZZ \subset\QQ$.
Note
\[
\alpha(s) = \alpha\left(\tfrac{a}{2^{n}}\right)=
\begin{cases}
\frac{a}{2^{n}\sqrt2}=\frac{a}{2^{n+1}}\sqrt2, & \text{when }\frac{a}{2^{n}}\le 2-\sqrt2,\\
\sqrt2\,\frac{a}{2^{n}}+1-\sqrt2
=1+\Big(\frac{a}{2^{n}}-1\Big)\sqrt2, & \text{when }\tfrac{a}{2^{n}}> 2-\sqrt2.
\end{cases}
\]
Hence the difference $\alpha(s)-\alpha(s')$ is one of: 
\begin{enumerate}
\item $\frac{a-b}{2^{n+1}}\sqrt{2}$ when $\frac{a}{2^{n}},\frac{b}{2^{n}}\le 2-\sqrt2$, 
\item $\frac{a-b}{2^{n}}\sqrt2$ when $\frac{a}{2^{n}},\frac{b}{2^{n}}> 2-\sqrt2$,
\item or takes form $-1+\Big(1+\frac{a}{2^{n+1}}-\frac{b}{2^{n}}\Big)\sqrt2$ when there is a mix of the prior cases.
\end{enumerate}
In cases (1) and (2), if $a\ne b$, the difference is irrational (nonzero $\sqrt2$-coefficient), contradicting our prior conclusion that $\alpha(s)-\alpha(s')\in \QQ$ when $\{2^{k}\alpha(s)\}=\{2^{k}\alpha(s')\}$; hence $a=b$ and $s=s'$.
In case (3), writing $\alpha(s)-\alpha(s')=A+B\sqrt2$ with $A=-1$, rationality forces $B=0$, i.e.
\[
1+\frac{a}{2^{n+1}}-\frac{b}{2^{n}}=0 \iff a-2b=-2^{n+1}.
\]
With $0\le a,b\le 2^{n}-1$ there is no solution, so case (3) cannot occur.
As such, we conclude $\alpha(s)-\alpha(s')\in\mathbb Q$, which implies $s=s'$. By the contrapositive, if $s\ne s'$ then $\{2^{k}\alpha(s)\}\ne \{2^{k}\alpha(s')\}$.
The exceptional case of $s=0,s'=1$ can be ruled out because $E_2(x)$ is defined mod 1, with $0$ and one identified, so $I_r = [0,1)\cap \DD_r$.
\end{proof}

\begin{lemma}\label{lem:telicPerturb}
Let $E_2:I\rightarrow I$ be the expanding map and let $n\in \ZZ^+$ and $m \geq 2$.
Put $c_1 = E_2^n(\alpha(s))$ for $s \in I_{n^m}$. There exists points $c_2,c_3 \in I$, with $c_1,c_2,c_3$ distinct, for which $|c_1-c_2|\leq O(2^{-n^m/2})$ and $|c_1-c_3|\leq O(2^{-n^m/2})$, such that
\begin{enumerate}
\item $c_2 \not\in E_2^n(\alpha(I_{n^m}))$, and
\item if $s' \in I_{n^m}$ is such that $E_2^n(\alpha(s'))=c_3$, then $|s-s'| \geq 1/3$.
\end{enumerate}
In addition, for every interval $J = [a, b] \subset I$ with $b-a = 1/2^n$, there are $O(2^{n^m-n})$ points $s \in I_{n^m}$ such that $E_2^n(\alpha(s)) \in J$.
\end{lemma}
\begin{proof}
Point (1) is an immediate consequence of the fact that $I_{n^m}$ is finite and totally disconnected.

Now we show (2): if $s' \in I_{n^m}$ is such that $E^n_2(\alpha(s')) =c_3$, then $|s-s'|\geq 1/3$.
Recalling \cref{lem:2nToOne}, $E_2^n$ is $2^n$-to-1, carrying each interval $[\frac{p}{2^n}, \frac{p+1}{2^n})$ onto $[0,1)$, $p=0,\dots,2^n-1$. Each such interval contains $2^{n^m-n}$ points of $I_{n^m}$ and $O(2^{n^m-n})$ points of the set $\alpha(I_{n^m})$. Since $\alpha$ is piecewise linear, the $O(2^{n^m-n})$ points are linearly distributed through all but one of the intervals $[\frac{p}{2^n}, \frac{p+1}{2^n})$.
It is assumed that $|c_1-c_3|\leq O(2^{-n^m/2})$. So, let $V$ be the interval $[c_1-\epsilon, c_1+\epsilon]$ for $\epsilon = O(2^{-n^m/2})$. But within any interval $[\frac{p}{2^n}, \frac{p+1}{2^n})$ there is a sub-interval $W$ of length $2\epsilon/2^n$ so that $E_2^n(W) = V$; that is, put $W = (E^n_2)^{-1}(V) \cap [\frac{p}{2^n}, \frac{p+1}{2^n})$ for some $p$.

There are then $O(2^{n^m-n})\cdot O(2^{-(n^m/2)-n+1}) = O(2^{1-2n + n^m/2})$ points in the set $\alpha(I_{n^m})$ contained in $W$. But $m\geq 2$ so this quantity is clearly asymptotically positive ($2^{1-2n + n^m/2}>0$ whenever $n\geq 4$). 
As such, there is always a $p \in [2^n-1]$ so that if $s \in [\frac{q}{2^n}, \frac{q+1}{2^n})$ for some $q\in [2^n-1]$, then $|s-p/2^n|\geq 1/3$. Then, we can take an $s' \in I_{n^m}\cap [\frac{p}{2^n}, \frac{p+1}{2^n})$ so that $s'\in \alpha(I_{n^m}) \cap W= (E^n_2)^{-1}(V) \cap [\frac{p}{2^n}, \frac{p+1}{2^n})$ where $V$ is as before. Finally, \cref{lem:smallPreimage} tells us that this $s'$ is unique---there is no other $s' \in I_{n^m}$ closer to $s$ satisfying the required property.

It is left to show that for every interval $J = [a, b] \subset I$ with $b-a = 1/2^n$, there are $O(2^{n^m-n})$ points $s \in I_{n^m}$ such that $E_2^n(\alpha(s)) \in J$.
Notice for $y \in I$,
\[
E_2^n(y) \in J \iff y \in U \coloneqq \bigcup_{k=0}^{2^n-1}\left(\frac{k}{2^n} + \frac{J}{2^n}\right),
\]
where $U$ is simply a union of $2^n$ disjoint intervals of length $|J|/2^n = 2^{-2n}$, so $|U|=2^n\cdot2^{-2n}=2^{-n}$. Put $A = \alpha^{-1}(U)$, a union of intervals with total length $|A|\leq \sqrt{2}|U| = \sqrt{2}\cdot2^{-n}$ which follows from the fact that
\[
\frac{1}{\sqrt{2}}|x-y| \leq |\alpha(x) - \alpha(y)|\leq \sqrt{2}|x-y|.
\]
So, the number of points of $I_{n^m}$ in $A$ is at most
\[
\frac{|A|}{2^{-n^m}} + 2 \leq \sqrt{2}\cdot 2^{n^m-n} + 2 = O\left(2^{n^m-n}\right).
\]
\end{proof}

In plain words, \cref{lem:telicPerturb} asserts that small variations in an instance of a telic problem, i.e.\ small perturbations in the location of the target point or interval, result in large deviations of the solution. This is stated explicitly in point (2) of the lemma: target points $c_1$ and $c_2$ are very close (within a factor of $2^{-n^m/2}$), but their solutions must be a distance at least $1/3$ apart.
In particular, this forces polynomials interpolating the solutions for different target points to have high degree.
An immediate consequence is that this forces any arithmetic circuits (with only $+,\times$ gates) solving $\ccRSTELIC_{(I, E_2)}(\alpha)$ to have depth $\Omega(n^2)$ (setting $m=2$ when applying \cref{lem:telicPerturb}).

Now we prove \cref{thmNoNaturalSearch}.

\begin{proof}[Proof of \cref{thmNoNaturalSearch}]
We first show there are dynamical systems $(I, F)$ and $(I, T)$ computable by BSS-machines for which $\ccRSTELIC_{(I, F)}(\Id)$ is not reducible to $\ccRSTELIC_{(I, T)}(\Id)$ by level 1 or 2 natural search reductions. We begin with the level 1 case. Let $F(x) = 2x\mod 1$ be the expanding map and put $T(x) = x^2$ a homeomorphism of the interval; these are both computable by BSS-machines.

Suppose for the sake of contradiction that $\ccRSTELIC_{(I, F)}(\Id)$ is reducible to $\ccRSTELIC_{(I, T)}(\Id)$ by level 1 natural search reductions.
Then there is a family of functions $\{\eta_n\}$ with $\eta_n:I \rightarrow I$ such that for all closed intervals $J \in \cI$ and all $x \in I_{n^2}$, $F^n(x) \in J \iff T^n(x) \in \eta_n(J)$. Letting $J$ be a point, we have $F^n(x) = J$ if and only if $T^n(x) = \eta_n(J)$ implies $\eta_n(F^n(x)) = T^n(x)$ for all $x \in I_{n^2}$. But notice $T$ is a bijection while $F$ is 2-1, and in particular $F$ maps dyadic rationals opposite $1/2$ to the same point: thus there are distinct points $x, y \in I_{n^2}$ for which $F(x)=F(y)$ while $T(x) \neq T(y)$, but this provides our contradiction since $\eta_n(F(x)) = \eta_n(F(y))$.

\emph{Level 2 natural search reductions.} Leaving $F$ and $T$ as the expanding map and contraction as above, we show that $\ccRSTELIC_{(I, F)}(\Id)$ is not reducible to $\ccRSTELIC_{(I, T)}(\Id)$ by level 2 natural search reductions. Again, proceed by contradiction, and suppose there exists a family of maps $\{\eta_n\}$ with $\eta_n:\cI\rightarrow \cI$ such that for all $x \in I$ and $J \in \cI$, $F^n(x) \in J$ if and only if $T^n(x) \in \eta_n(J)$. Since we quantify $\forall x\in I$, the assumption implies the pullback identity $(F^n)^{-1}(J) = (T^n)^{-1}(\eta_n(J))$. But this equality cannot hold since $F$ is 2-1, and, in particular for any $J \in \cI$, $(F^n)^{-1}(J)$ is disconnected, while $(T^n)^{-1}(\eta_n(J))$ is connected since $\eta_n:\cI\rightarrow \cI$ and $T$ is a homeomorphism of $I$. Hence the sets $(F^n)^{-1}(J)$ and $(T^n)^{-1}(\eta_n(J))$ cannot be equal for any $J \in \cI$ and all sufficiently large $n$.

\emph{Level 3 natural search reductions.} It is left to show existence of dynamical systems admitting algebraic search telic problems for which one is not reducible to the other by level 3 natural search reductions. This claim follows for reasons less elementary than the previous cases, and relies on \cref{lem:telicPerturb}.
Continue to identify $(I, F)$ with the expanding map, and now let $(I, T)$ be the trivial dynamical system with $T = \Id$. We aim to show there does not exist a level 3 natural search reduction from $\ccRSTELIC_{(I, F)}(\alpha)$ to $\ccRSTELIC_{(I, T)}(\Id)$, where $\alpha$ is the function defined previously.

Assume for the sake of contradiction that there exists a function family $\{\eta_n\}$ with $\eta_n:\cI\rightarrow \cI$ such that for every $x \in I_{n^2}$ and for all $J \in \cI$, $F^n(\alpha(x)) \in J$ if and only if $T^n(x) \in \eta_n(J)$. Consider the class of instances of $\ccRSTELIC_{(I, F)}(\alpha)$ for which $J$ is an interval of length $1/2^n$. By point (2) and the later assertion of \cref{lem:telicPerturb}, there are $O(2^{n^2-n})$ possible solutions and there is at least one possible solution in each interval $[\frac{p}{2^n}, \frac{p+1}{2^n})$, $p=0,1,\dots,2^n-1$. But recall level 3 natural search reductions require that $F^n(\alpha(x)) \in J$ implies $T^n(x) \in \eta_n(J)$ for every $x \in I_{n^2}$. By \cref{lem:telicPerturb}, this forces us to conclude $|\eta_n(J)|\geq 1/3$, because there is a solution in every interval $[\frac{p}{2^n}, \frac{p+1}{2^n})$, all of which $\eta_n(J)$ must contain. But this implies there are points contained in $\eta_n(J) \cap I_{n^2}$ that are not contained in $(F^n\circ \alpha)^{-1}(J) \cap I_{n^2}$. This provides a contradiction because this is not a level 3 natural search reduction: a solution to an instance of $\ccRSTELIC_{(I, T)}(\Id)$ does not imply a solution to the corresponding instance of $\ccRSTELIC_{(I, F)}(\alpha)$.

We conclude by noticing that the expanding map is a ``chaotic" dynamical system---it has positive topological entropy in particular, while the contraction map $T(x)=x^2$ and the trivial system $T = \Id$ are highly regular---they have zero topological entropy. In particular, it is not hard to see that for such regular systems $\ccRSTELIC_{(I, T)}(\hat{g}) \in \ccP_\RR$---see e.g.\ \cref{propRotationInP} for proof.
\end{proof}

\begin{remark}
The proofs of \cref{thmRulingOutNaturalTelic,thmNoNaturalSearch} were rather elementary, and did not use any sophisticated results from dynamical systems. They demonstrated the elementary theory and structure of standard dynamical systems of the unit interval suffices for ruling out entire classes of nontrivial reductions between telic problems that would \emph{a priori} seem hard to exclude. The upshot is that there is a wealth of far more sophisticated dynamical structure available for use in ruling out more complex classes of reductions, such as certain level 4 natural search reductions. In order to manage the scope of this work we leave the further development of this line of inquiry to future work.
\end{remark}

\subsection{Placing structure on level 4 natural search reductions}\label{secLevel4}

For $a \geq 0$, identify $\cI_a$ with the line segment
\[
\{(c, a) \in \RR^2 : a \leq c \leq 1-a\} \subset \RR^2,
\]
so that each element $(c, a) \in \cI_a$ describes a closed interval for which $c$ is the midpoint and $a$ the radius: i.e.\ $[c-a, c+a]$.
We abuse notation in the same way as before, treating $\eta_n:\cI_a\rightarrow \cI_a$ as a function of the set $\cI_a \subset \RR^2$. For $J=(c,a) \in \cI_a$ we write $x \in J$ to mean $x \in [c-a,c+a]$.

For the purpose of this section we shall consider a generalization of algebraic search telic problems, called \emph{bounded algebraic search telic problems}.
As before, let $g:I \rightarrow I$ be a homeomorphism of $I$ computable by a BSS-machine.
Let $\cU \coloneqq \cup_{n\geq 1}\{\One_n\}$, and $\ccK:\cU\rightarrow \cU$ be a BSS-machine halting in polynomial time. If $\ccK(\One_n)= \One_l$, define $k:\NN\rightarrow \NN$ to be defined so that $k(n) = l$.
Let $\ccA:\cU\rightarrow \RR^+$ denote a polynomial-time BSS-machine returning a positive real value $r$ on input $\One_n$. Put $a(n)\coloneqq \ccA(\textbf{1}_n) \in \RR^+$.
Recall $\cB\subset \RR^\infty$ is the set $\cB = \cup_{n=1}^\infty \cB^{(n)}$ where 
\[
\cB^{(n)} = [0,1]\times [0,1] \times \prod_{j=1}^n \{1\}.
\]
Define
\[
\cA^{(n)}_{a(n)} = \left\{(c, r, 1,\dots,1) \in \cB^{(n)} : (c, r) \in \cI_{a(n)} \text{ and } \exists s \in I_{n^2} \text{ s.t. } F^{k(n)}(g(s)) \in [c-r, c+r]\right\},
\]
and put $\cA_{a} = \cup_{n=1}^\infty \cA^{(n)}_{a(n)}$.
This gives us a structured decision problem
\[
\ccRBTELIC_{(I, F)}(a,k,g)\coloneqq (\cB, \cA_a)
\]
we call a \emph{bounded algebraic telic problem}. It is easy to check that these are contained in $\ccNP_\RR$ as in \cref{lemRealTelicNP}.
Write $\ccRSBTELIC_{(I, F)}(a,k,g)\subset \RR^\infty \times \RR^\infty$ to denote the search version of the bounded algebraic telic problem $\ccRBTELIC_{(I, F)}(a,k,g)$.

Natural search reductions of any level carry over to the setting of bounded algebraic telic problems in the natural way with the only alteration being that the maps $\eta_n$ are defined over the domain $\cI_{a(n)}$ rather than $\cI$; that is $\eta_n:\cI_{a(n)}\rightarrow \cI_{a(n)}$.

We use a fixed point argument to obtain the following theorem forcing rather strong irregularity on reductions between certain bounded algebraic search telic problems.

\begin{theorem}\label{thm:fixedPoint}
There exists a BSS-computable dynamical system $(I, F)$, and functions $a(n), k(n)$ for which $\ccRSBTELIC_{(I, F)}(a, k,\Id)$ is not reducible to the trivial search problem  $\ccRSBTELIC_{(I, \Id)}(a, \Id,\Id)$ coming from $(I, \Id)$ by level 4 natural search reductions whose functions $\eta_n:\cI_{a(n)}\rightarrow \cI_{a(n)}$ composing the sequence $\{\eta_n\}_{n=1}^\infty$ satisfy any of the following properties:
\begin{enumerate}
\item continuity,
\item non-expansivity,
\item there is a lower semi-continuous function $\varphi:\RR^2\rightarrow [0,\infty)$ such that
\[
\rho(z, \eta_n(z))\leq \varphi(z)-\varphi(\eta_n(z))
\]
for all $z \in \cI_{a(n)}$.
\end{enumerate}
\end{theorem}
\begin{proof}
Define $(I, F)$ to be a BSS-computable rigid circle rotation by any irrational value $\kappa \in \RR^+$, so $F(x) = x + \kappa \mod 1$.
Let $\ccRSBTELIC_{(I, F)}(a,k,\Id)\subset \RR^\infty \times \RR^\infty$ be the search version of the bounded algebraic telic problem $\ccRBTELIC_{(I, F)}(a,k,\Id)$, with $a(n)\coloneqq \ccA(\textbf{1}_n)$ defined to be $a(n) =1/2^{n^2}$---this implies the intervals in $\cI_{a(n)}$ have length $1/2^{n^2}$.
Define $\ccK(\textbf{1}_n)$, treated as a positive integer $k(n)$, to be
\[
k(n) = \begin{cases}
n &\text{ when } \rho(F^{k(n)}(0), 0) > \frac{1}{2^{n^2}}\\
n+1 &\text{ otherwise.}
\end{cases}
\]
Because $F$ is a rigid rotation by fixed $\kappa>0$, there exists an $N \in \NN$ such that for all $n \geq N$, if
\[
\rho(F^{n}(0), 0) \leq \frac{1}{2^{n^2}}, \;\text{ then }\; \rho(F^{n+1}(0), 0) > \frac{1}{2^{n^2}}.
\]
In addition, the machines $\ccA$ and $\ccK$ computing $a(n)$ and $k(n)$ both clearly run in polynomial-time.

For any value $0\leq a\leq 1/2$, $\cI_a$ is compact, nonempty, and convex.
Hence, if $\eta_n:\cI_{a(n)}\rightarrow \cI_{a(n)}$ is either (1) continuous, (2) non-expansive, or (3) there is a lower semi-continuous function $\varphi:\RR^2\rightarrow [0,\infty)$ such that
\[
\rho(z, \eta_n(z))\leq \varphi(z)-\varphi(\eta_n(z)) \text{ for all $z \in \cI_{a(n)}$},
\]
then it follows by (1) Brouwer's Fixed Point Theorem, (2) Browder-Kirk's Fixed Point Theorem, and (3) Caristi's Fixed Point Theorem, that there is a $\bar{x} \in \cI_{a(n)}$ such that $\eta_n(\bar{x}) = \bar{x}$ (see \cite{pata2019fixed} for a nice source on fixed point theorems).

However, the definitions of functions $a(n), k(n)$ prevent any of the $\eta_n$ from having a fixed point for sufficiently large $n$.
To see this, first notice that by the definition that $a(n) = 1/2^{n^2}$, the length of any interval in $\cI_{a(n)}$ is $1/2^{n^2+1}$, and hence any such intervals contains at least one element of the $n^2$-discretization of $I$.
Thus there is always a solution.
In a level 4 natural search reduction from $\ccRSBTELIC_{(I, F)}(a, k,\Id)$ to $\ccRSBTELIC_{(I, \Id)}(a, \Id,\Id)$, $\eta_n$ must satisfy, for every $n\in \NN$, $x \in I_{n^2}$ and $J \in \cI_{a(n)}$, the property that $x \in \eta_n(J)$ implies $F^{(k(n)}(x) \in J$.

Consequently, for sufficiently large $n$, if $\eta_n:\cI_{a(n)}\rightarrow \cI_{a(n)}$ has a fixed point, there is a $J \in \cI_{a(n)}$ such that $\eta_n(J) = J$, and in the level 4 natural search reduction we have $x \in \eta_n(J)=J$ implies $F^{k(n)}(x) \in J$.
But by definition of $k(n)$, for all sufficiently large $n$ we have $F^{k(n)}(J) \cap J = \emptyset$, so $x \in \eta_n(J)=J$ does not imply $F^{k(n)}(x) \in J$. We conclude that function families $\{\eta_n\}$ for which the $\eta_n$ have a fixed point fail to serve as level 4 natural search reductions.
\end{proof}

By condition (1) of \cref{thm:fixedPoint}, $\eta_n$ cannot be continuous. This implies the $\eta_n$ cannot be computed by arithmetic circuits, for instance, which use only $+$ and $\times$ gates. In fact the lack of continuity is a rather strong condition, and implies the functions $\eta_n$ must already possess non-trivial behavior. Indeed, this can be seen through applying the Lefschetz Fixed Point Theorem, and noticing that if an $\eta_n$ cannot have a fixed point, it must annihilate certain homology groups, which implies severe singularities or topological folding. We leave further commentary to this end to future work, although we will remark that this line of attack appears to be vulnerable to proving lower-bounds on the complexity of the $\eta_n$ through the application of homological and homotopy theoretic machinery.

\begin{remark}
The example of the rigid circle rotation used to prove \cref{thm:fixedPoint} is quite simple and also rather contrived, so the proof works for basic reasons. However, because the example considered in the proof is already so simple, the upshot is that this result readily carries over to a host of other cases with far less trivial properties. In particular, one can expect that most level 4 natural search reductions between highly distinct systems cannot satisfy any of the properties (1)--(3) of \cref{thm:fixedPoint}. One way to force such behavior easily for systems more complex then rigid rotations, is to generalize the $\ccK$ so that it is now a function of both $n$ and the target interval $J$, defined so that $k(n)\neq n$ when the map $F:I\rightarrow I$ would carry a point of $I_{n^2}$ contained in $J$ back into $J$ after $n$ iterations.
\end{remark}

\subsection{A concluding comment on Level 4a reductions}

In constructing a sequence of functions $\{\eta_n\}_{n=1}^\infty$ for reducing any search telic problem $\ccRSTELIC_{(I, F)}(g)$ to any other $\ccRSTELIC_{(I, T)}(\hat{g})$ by level 4 natural search reductions, the natural approach would be to find a sequence of functions that ``extends" nicely to the rest of state space, in that instead of  $\forall x \in I_{n^2}$, we take $\forall x \in I$: we say $\ccRSTELIC_{(I, F)}(g)$ is reducible to $\ccRSTELIC_{(I, T)}(\hat{g})$ by \emph{level 4a natural search reductions} if there is a family of functions $\{\eta_n\}_{n=1}^\infty$ with $\eta_n:\cI\rightarrow \cI$, such that for all $J \in \cI$ and for all $x \in I$, $T^n(\hat{g}(x)) \in \eta_n(J)$ implies $F^n(g(x)) \in J$, and $T^n(\hat{g}(I_{n^2})) \cap \eta_n(J) = \emptyset$ implies $F^n(g(I_{n^2})) \cap J = \emptyset$ (the later condition set for the case when there is no solution).

Studying level 4a natural search reductions is a natural first step in attempting to understand the structure of level 4 natural search reductions.
As a consequence, we collect some basic properties forced upon function sequences composing level 4a natural search reductions, which may serve as useful observations in further development.

Let $(I, \cS, \lambda, F)$ and $(I, \cS, \lambda, T)$ be measure preserving  systems of the unit interval, both preserving the Lebesgue measure $\lambda$, and both computable by BSS-machines.
Suppose there exists a level 4a natural search reduction from $\ccRSTELIC_{(I, F)}(\Id)$ to $\ccRSTELIC_{(I, T)}(\Id)$.
Then the hypothesis $T^n(x) \in \eta_n(J) \implies F^n(x) \in J$ for every $x \in I$ and $J \in \cI$ gives us the identify $(T^n)^{-1}(\eta_n(J)) \subseteq (F^n)^{-1}(J)$. We immediately obtain some measure-theoretic consequences, valid for every $n$. First,
\[
\lambda(\eta_n(J)) = \lambda((T^n)^{-1}(\eta_n(J))) \leq \lambda ((F^n)^{-1}(J)) = \lambda(J).
\]
Thus, if $J \in \cI$ is a point, $\lambda(\eta_n(J)) = 0$, and hence $\eta_n(J)$ is also a point because $\eta_n$ maps intervals to intervals, and the only measure zero intervals are points (degenerate intervals). Furthermore, because $\lambda$ measures lengths of intervals on $I$, we have $|\eta_n(J)|\leq |J|$ for all $J \in \cI$.
We also see that disjointness is preserved, modulo null sets: if $J_1, J_2 \in \cI$ are disjoint, $J_1 \cap J_2 = \emptyset$, then $(F^n)^{-1}(J_1) \cap (F^n)^{-1}(J_2) = \emptyset$. This gives us
\[
(T^n)^{-1}(\eta_n(J_1)) \cap (T^n)^{-1}(\eta_n(J_2)) = \emptyset.
\]
If $T^n$ is onto mod measure-zero sets, we conclude $\lambda(\eta_n(J_1) \cap \eta_n(J_2)) = 0$; images of disjoint intervals are (Lebesgue) disjoint.
In particular, if $\{J_k\}$ is a finite family of disjoint elements of $\cI$, then $\sum_k\lambda(\eta_n(J_k)) \leq \sum_k \lambda(J_k)$:

\begin{align*}
\sum_k\lambda(\eta_n(J_k)) = \sum_k \lambda((T^n)^{-1}(\eta_n(J_k))) &= \lambda \left( \bigsqcup_k(T^n)^{-1}(\eta_n(J_k))\right)\\
&\leq \lambda\left(\bigsqcup_k(F^n)^{-1}(J_k)\right) = \sum_k \lambda(J_k).
\end{align*}

One useful result these comments provide is the following

\begin{proposition}\label{propLevel4aRegularity}
Let $(I, F)$ be any BSS-computable dynamical system for which $F$ is many-to-one, and let $(I, T)$ be a BSS-computable dynamical system with $T$ a homeomorphism of $I$. Suppose $\ccRSTELIC_{(I, F)}(g)$ is a algebraic search telic problem reducible to $\ccRSTELIC_{(I, T)}(\hat{g})$ by level 4a natural search reductions. Then if $\{\eta_n\}_{n=1}^\infty$ is the corresponding function sequence, every $\eta_n$ carries points (degenerate closed intervals) in $\cI$ to points in $\cI$; that is, $\eta_n:I\rightarrow I$.
\end{proposition}
\begin{proof}
Recall the hypothesis $T^n(x) \in \eta_n(J) \implies F^n(x) \in J$ for every $x \in I$ and $J \in \cI$ gives us the identify $(T^n)^{-1}(\eta_n(J)) \subseteq (F^n)^{-1}(J)$.
Take $J = y \in \cI$ to be a point.
The functions $T, \hat{g}$ are taken to be homeomorphisms, which implies $(T^n\circ \hat{g})^{-1}(\eta_n(y))$ is a point contained in the preimage $(F^n \circ g)^{-1}(y)$ using the many-to-one assumption on $F$.
Then $(T^n\circ \hat{g})^{-1}(\eta_n(y)) \in (F^n \circ g)^{-1}(y)$. In particular, we conclude that $\eta_n(y)$ carries points in $\cI$ to points in $\cI$, i.e. $\eta_n:I\rightarrow I$ for all $n \in \NN$.
\end{proof}

\cref{propLevel4aRegularity} forces useful regularity on the functions $\eta_n:\cI\rightarrow \cI$. Identify $\cI$ with the closed triangle $\Delta = \{(a, b) \in I^2: 0\leq a\leq b\leq 1\}$, and write $\eta_n(J) = \eta(a, b) = [\gamma_1(a, b), \gamma_2(a, b)] \in \cI$. Then \cref{propLevel4aRegularity} forces $\gamma_1$ and $\gamma_2$ to coincide on the diagonal of $\Delta$, that is $\gamma_1(x, x) = \gamma_2(x,x)$.

\appendix

\section{Ruling out reductions in a general setting}\label{secGeneralReductions}

The proof of \cref{thmRulingOutNaturalTelic} fixed two simple dynamical systems to prove the assertion. Because of their simplicity, the result can be easily extended using identical argument to more complex dynamical systems of the unit interval.
However, one may ask the extent to which reductions can be ruled out between dynamical systems of a more general class. In this section we show that this can be done.

We begin by breifly providing a generalization of algebraic telic problems to arbitrary BSS-computable systems.

\begin{definition}[Generalized algebraic telic problems]\label{defGenRtelic}
Let $X \subset \RR^d$, and $(X, T)$ be a dynamical system computable by a BSS-machine $\ccM$. Let $\cB\subset \RR^\infty$ be the set $\cB = \cup_{n=1}^\infty \cB^{(n)}$ where 
\[
\cB^{(n)} = X \times \RR_{\geq 0} \times \prod_{j=1}^n \{1\}.
\]
Let $g:X \rightarrow X$ be a homeomorphism of $X$ computable by a BSS-machine.
Let
\[
\cA^{(n)} = \left\{(w, \delta, 1,\dots,1) \in \cB^{(n)} : \exists s \in I_{n^2} \text{ s.t. } T^n(g(s)) \in B_\delta(w)\right\},
\]
and put $\cA = \cup_{n=1}^\infty \cA^{(n)}$.
A \emph{generalized algebraic telic problem} coming from a BSS-computable dynamical system $(X, T)$, is the structured decision problem $\ccGRTELIC_{(X, T)}(g) \coloneqq (\cB, \cA)$.
\end{definition}

Let $X \subseteq \RR^d$ and $Y \subseteq \RR^m$, and let $(X,T)$ and $(Y, S)$ be BSS-computable topological dynamical systems admitting generalized algebraic telic problems $\ccGRTELIC_{(X, T)}(g)$ and $\ccGRTELIC_{(Y, S)}(\hat{g})$, respectively.
We shall say $(X, T)$ is reducible to $(Y, S)$ by \emph{natural telic reductions} if there exists a sequence of functions $\{\varphi_i\}_{i=1}^\infty$, $\varphi_i:X\rightarrow Y$, each with the property that for any $x \in X$ and closed ball $B \subset X$, $T^{n}(x) \in B$ implies $S^{n}(\varphi_n(x)) \in \varphi_n(B)$.
If a natural telic reduction exists between the systems $(X, T)$ and $(Y, S)$, and the sequence of functions $\{\varphi_n\}$ is uniformly and polynomial-time computable by BSS-machines, then they can be used in reducing instances of a generalized algebraic telic problem $\ccGRTELIC_{(X, T)}(g)$ coming from $(X, T)$ to a telic problem $\ccGRTELIC_{(Y, S)}(\hat{g})$ coming from $(Y, S)$.

Recall a dynamical system $(X, T)$ is \emph{minimal} if it has no proper, closed $T$-invariant subsets.
We have the following

\begin{theorem}\label{thmImpossibleRed}
Let $X \subset \RR^d$, and $Y\subset \RR^m$ be infinite. Let $(X, T)$ and $(Y, S)$ be BSS-computable dynamical systems. Whenever $(X, T)$ has a periodic point and $(Y, S)$ is minimal, $(X, T)$ is not reducible to $(Y, S)$ by natural telic reductions.
\end{theorem}

Before giving proof of \cref{thmImpossibleRed} we provide some lemmas.
For these, recall $\Per_n(X, T)$ is the subset of $X$ consisting of all periodic points of period $n$, and $\Per(X, T) = \cup_{n\geq 1} \Per_n(X, T)$. We also use the fact that if $\varphi:X \rightarrow Y$ is a function satisfying $\varphi\circ T = S \circ \varphi$, this implies $\varphi\circ T^n = S^n \circ \varphi$ for all $n \in \NN$. This follows from an easy inductive argument, where at the inductive step use $\varphi\circ T^{n+1} = \varphi \circ T \circ T^n$ to get $\varphi \circ T^{n+1} = (S \circ \varphi) \circ T^n = S \circ (\varphi \circ T^n)$, so applying the induction hypothesis that $\varphi \circ T^n = S^n \circ \varphi$, we obtain $S \circ (\varphi \circ T^n) = S \circ (S^n \circ \varphi) = S^{n+1} \circ \varphi$ which proves the assertion.

\begin{lemma}\label{lemSemiConjPer}
Let $(X, T)$ and $(Y,S)$ be dynamical systems, and let $\varphi:X\rightarrow Y$ be a function such that $\varphi\circ T = S \circ \varphi$. Then $\varphi(\Per_n(X, T)) \subseteq \Per_n(Y,S)$ for every $n \geq 1$, and $\varphi(\Per(X, T)) \subseteq \Per(Y, S)$.
\end{lemma}
\begin{proof}
If $x \in \Per_n(X, T)$ then by the assumption that $\varphi\circ T = S \circ \varphi$, we have $S^n(\varphi(x)) = \varphi(T^n(x)) = \varphi(x)$, so $\varphi(x) \in \Per_n(Y,S)$. Using this, we see that
\begin{align*}
\varphi(\Per(X, T)) &= \varphi(\cup_{n\geq 1} \Per_n(X, T))\\
&= \cup_{n \geq 1}\varphi(\Per_n(X, T)) \\
&\subseteq \cup_{n\geq 1}\Per_n(Y,S)\\
&= \Per(Y,S).
\end{align*}
\end{proof}

\begin{lemma}\label{lemConjEquiv}
Let $(X, T)$, $(Y, S)$ be dynamical systems and suppose there exists a function $\varphi:X\rightarrow Y$ such that for all $x \in X$, and any subset $J \subseteq X$ of state space, $T^n(x) \in J$ implies $S^n(\varphi(x)) \in \varphi(J)$. Then $\varphi\circ T = S \circ \varphi$.
\end{lemma}
\begin{proof}
Take $J = z\in X$. Then $T^n(x) = z$ implies $S^n(\varphi(x)) = \varphi(z)$ so $\varphi(T^n(x)) = S^n(\varphi(x))$, i.e.\ $\varphi\circ T^n = S^n \circ \varphi$.
\end{proof}

\begin{lemma}\label{lemInfiniteMinimal}
Let $(Y, S)$ be a topological dynamical system with $Y \subseteq \RR^d$ infinite. Then if $(Y, S)$ is minimal the system does not have any periodic orbits.
\end{lemma}
\begin{proof}
Suppose for the sake of contradiction that $y\in Y$ is periodic. Then $S^k(y)=y$ for some $k\ge1$, so the orbit
\[
\mathcal{O}(y)=\{y,S(y),\dots,S^{k-1}(y)\}
\]
is finite. Since $Y$ is a subspace of the metric space $\mathbb{R}^d$, $\overline{\mathcal{O}(y)}=\mathcal{O}(y)$. Moreover $\mathcal{O}(y)$ is $S$-invariant and nonempty. Because $Y$ is infinite, $\mathcal{O}(y)\neq Y$. Thus $\mathcal{O}(y)$ is a nonempty proper closed $S$-invariant subset of $Y$, contradicting minimality. Therefore no periodic point exists.
\end{proof}

We now prove \cref{thmImpossibleRed}.

\begin{proof}[Proof of \cref{thmImpossibleRed}]
Take $(X, T)$ to be any BSS-computable dynamical system with at least one periodic orbit, and take $(Y, S)$ to be any BSS-computable minimal system, such as a rigid circle rotation by an irrational value.
Proceed by contradiction and suppose there exists a natural telic reduction between the two systems.

Then there exists a function family $\{\varphi_n\}_{n=1}^\infty$ with $\varphi_n:X\rightarrow Y$ satisfying the property that for any $x \in X$ and $J\subseteq X$, $T^n(x) \in J$ implies $S^n(\varphi_n(x)) \in \varphi_n(J)$.
But by \cref{lemConjEquiv} this implies each $\varphi_n$ satisfies $\varphi_n\circ T^n = S^n \circ \varphi_n$.
Then, \cref{lemSemiConjPer} tells us $\varphi_n(\Per(X, T)) \subseteq \Per(Y,S)$. But $\text{card}(\Per(X,T)) >0$, while $\text{card}(\Per(Y,S)) = 0$ since $\Per(Y,S)$ is empty by \cref{lemInfiniteMinimal}. Hence $\varphi_n(\Per(X,T)) \subseteq \Per(Y,S)$ is an absurdity and we obtain our contradiction.
\end{proof}

\bibliographystyle{abbrv}
\bibliography{references}

\end{document}